\documentclass[american]{article}

\usepackage[T1]{fontenc}
\usepackage[utf8]{inputenc}
\usepackage{babel}
\usepackage{csquotes}

\usepackage{amsmath,amssymb,amsthm}
\usepackage{float}
\usepackage{geometry}
\usepackage{setspace}
\usepackage{slashed}
\usepackage{dsfont}
\usepackage{algpseudocode}
\usepackage{tikz}
\usetikzlibrary{backgrounds}

\floatstyle{ruled}
\newfloat{algorithm}{tbp}{loa}
\floatname{algorithm}{Algorithm}

\theoremstyle{definition}
\newtheorem{example}{Example}
\newtheorem{defn}{Definition}

\theoremstyle{plain}
\newtheorem{prop}{Proposition}
\newtheorem{thm}{Theorem}
\newtheorem*{thm*}{Theorem}
\newtheorem{lem}{Lemma}

\theoremstyle{remark}
\newtheorem{rem}{Remark}

\usepackage[style=apa, backend=biber]{biblatex}
\usepackage[
  bookmarks=false,
  breaklinks=false,
  pdfborder={0 0 1},
  colorlinks=true,
  linkcolor=blue,
  citecolor=blue,
  urlcolor=blue
]{hyperref}

\begin{document}
\title{Random Matching with Minimums}
\author{Will Sandholtz\thanks{UC Berkeley. Contact: \protect\href{mailto:willsandholtz@berkeley.edu}{willsandholtz@berkeley.edu}.}
  \and Andrew Tai\thanks{US Dept. of War. Contact: \protect\href{mailto:andrew.a.tai@gmail.com}{andrew.a.tai@gmail.com}.
    Disclaimer: The views expressed are those of the author and do not
    reflect the official policy or position of the Department of War or
    the U.S. Government.}}
\date{May 2026}
\maketitle
\begin{abstract}
  We study stochastic object assignment problems in which objects may
  have minimum and maximum requirements, such as with classes with upper
  and lower enrollment bounds. We construct a new random assignment
  mechanism, the minimums probabilistic serial (MPS) mechanism, which
  generalizes the Probabilistic Serial mechanism of \textcite{BM01}. The
  random allocation produced by MPS is guaranteed to be Pareto efficient;
  that is, there is no other implementable allocation that all agents
  prefer via first order stochastic dominance. We also show that MPS
  is i) envy-free, in that no agent will strictly prefer another agent's
  assignment, and ii) weak strategyproof, in that agents cannot achieve
  a better assignment by misreporting their preferences.
\end{abstract}

\section{Introduction}

In the Book of Numbers, Moses casts lots to divide territory among
the tribes of Israel -- thus object allocation via lottery goes at
least as far back as the Hebrew Bible. In modern settings, public
housing, dorm allocation, public school allocation, jury selection,
and even immigration are examples of ``object'' allocation via randomization.
When the objects vary significantly but monetary transfers are infeasible,
randomization can impose ex ante fairness.

Going further, social choice mechanisms attempt to incorporate the
preferences of the agents to achieve efficiency. The simplest example,
random serial dictatorship (RSD) randomizes the order in which agents
select their favored objects. \textcite{BM01} point out the efficiency
shortcomings of RSD (which we will recount later) and propose the
probabilistic serial mechanism (PS). In PS, agents ``eat'' probability
of their favored object until it is exhausted, yielding a lottery
over the objects. Once agents have lotteries over objects, the Birkhoff-von
Neumann theorem guarantees that it can be translated to a lottery
over deterministic allocations.

Often the objects assigned are bound by minimum constraints in addition
to capacities. For example, consider assigning projects to workers;
a project may require a minimum number of workers to be successful.
Another example is the school club system in Japanese secondary schools;
clubs are often mandatory and incorporated as class credit. A club
may have a maximum capacity of students, and it may also have a minimum
constraint. For example, sports teams are often organized as clubs
-- the basketball team requires at least 5 students.

This paper considers randomized object allocation with minimums. While
random serial dictatorship\footnote{With proper modifications to ensure respecting the constraints; see
  \textcite{MT13}.} suffices to select an ex post efficient outcome, it inherits the
same problems from \textcite{BM01}'s setting; the lotteries are not
efficient. We therefore generalize the probabilistic serial mechanism
to accommodate minimums and maximums while maintaining efficiency
and fairness.

Towards a formal model, let the set of agents be $N$, each of whom
demand $d$ distinct objects in $O$. A deterministic allocation is
denoted $M=[M_{ij}]_{i\in N,j\in O}$, which is a matrix of 0s and
1s.\footnote{Note it is not bistochastic.} A given object $o_{j}$
has a maximum capacity $c_{j}$ and a minimum requirement $m_{j}$.
The final allocation requires $m_{j}\leq\sum_{i}M_{ij}\leq c_{j}$.
We seek a matrix $\mu=[\mu_{ij}]$ where row $i$ represents a lottery
over objects for agent $i$. The advantage of working explicitly with
lotteries is that they can be more efficient (in a precise sense discussed
in the next section). Of course, it is critical and nontrivial to
ensure that $\mu$ can indeed be implemented via a lottery over feasible
deterministic allocations.

Thus we have two tasks. The main task is to produce an efficient random
allocation $\mu$. To do so, we characterize the set of implementable
random allocations and paths to them from the origin. We express this
geometrically as a polytope represented by linear inequalities $\{\mu:A\mu\leq b\}$.
Then we apply results from \textcite{Balbuzanov22} to construct our
generalization of the probabilistic serial mechanism, which we call
\emph{Minimums Probabilistic Serial (MPS)}. Once we have the characterization
$\{\mu:A\mu\leq b\}$, we apply \textcite{Balbuzanov22}. MPS is guaranteed
to produce a Pareto efficient random allocation, in that there is
no other implementable random allocation that all agents prefer via
first order stochastic dominance. We further show that MPS is envy
free, meaning no agent will strictly prefer another agent's allocation.
In the special case of unit demand where $d=1$, MPS is weak strategyproof,
meaning a misreport cannot result in a stochastically dominant allocation.
While efficiency is inherited as a special case from \textcite{Balbuzanov22},
the other properties are not.

The second task is to implement $\mu$; that is, to write $\mu$ as
a convex combination of allowable deterministic allocations. Since
minimum and capacity constraints are bihierarchical in the sense of
\textcite{BCKM13}, we can directly apply their implementation results.

\textcite{BM01} introduce the probabilistic serial (PS) mechanism, the
first eating algorithm for random allocations. They show it is efficient
and weak strategyproof for the special case of 1:1 matching problems.
\textcite{BCKM13} (henceforth BCKM13) deal with implementation in many-to-many
random matching problems with bihierarchical constraints. They also
present a generalization of PS, termed \emph{Generalized Probabilistic
  Serial} (GPS), when the bihierarchical constraints are only maximums.
However, minimum quotas on objects cannot be represented with bihierarchical
constraints using only maximum capacities. With three or more objects,
at least two of which have nonzero minimums, the bihierarchical constraint
structure cannot be maintained. Appendix \ref{sec:BCKM} contains
details and an example.

\textcite{Balbuzanov22} (henceforth B22) describes the \emph{Generalized
  Constrained Probabilistic Serial (GCPS)} mechanism, which generalizes
the setting to allow arbitrary constraints. His primitive is an enumeration
of allowable deterministic allocations. With this list, the \emph{facet
  enumeration problem (FEP)} is applied to characterize the polytope
of implementable random allocations and paths to them. This characterization
is then fed into Generalized Constrained Probabilistic Serial (GCPS).
However, it is nontrivial to convert even simple constraints to an
enumeration of deterministic allocations; this is a factorial complexity
problem, as Example \ref{exa:factorial} notes.
\begin{example}
  \label{exa:factorial} Let there be $n$ agents in $N$ and $k$ objects
  in $O$, with $n<k$. Let $m_{j}=0$ and $c_{j}=1$ for all $j\in O$.
  Then there are $\frac{k!}{(k-n)!}=n!\binom{k}{n}$ different allowable
  deterministic allocations.
\end{example}
Further, FEP procedures have worst-case exponential complexity in
the inputs and outputs (which are proportional to the factorial number
of deterministic allocations). Therefore the procedure in B22 quickly
becomes intractable even for problems of modest size.

To address the gaps in the literature, we take minimum and capacity
constraints as our primitives and explicitly describe our mechanism,
MPS. In the unit demand case, we also provide a polynomial-time algorithm.
By providing the first feasible algorithm for random matching with
minimums, we hope that our results help to connect the literature
and implementation in real applications.

In addition to BCKM13 and B22, many have worked on random allocation.
\textcite{KS06} extend Probabilistic Serial to indifferences; \textcite{Yilmaz09}
also considers endowments. \textcite{Kojima09} deals with multi-unit
demand. \textcite{Heo14} and \textcite{HHKKU14} offer characterizations
in this setting. Papers including \textcite{CK10}, \textcite{KM10}, and
\textcite{LP16} deal with probabilistic allocations in large markets.
A number of papers deal with constraints such as minimums in matching.
A non-exhaustive list includes \textcite{BFIM10}, \textcite{HIM11}, and
\textcite{EHYY14}.

\section{Model}

\label{sec:Model}

Let $N=\{1,\ldots,N\}$ be the set of agents and $O=\{o_{1},\ldots,o_{O}\}$
be the set of objects. Each agent demands exactly $d\in\mathbb{N}$
distinct objects. We use $j\in O$ to refer a generic object. A \textbf{deterministic
  allocation} is a matrix $M\in\{0,1\}^{N\times O}$, where $M_{ij}=1$
means that agent $i$ is assigned to object $j$. We require:
\begin{align*}
  \sum_{j}M_{ij} & =d & \forall i\in N &  & \tag{Demand}
\end{align*}
That is, each agent is assigned to exactly $d\in\mathbb{N}$ objects,
and at most one of each. Each object $j$ has two characteristics:
a maximum capacity of agents, $c_{j}\in\mathbb{N}$, and a minimum
required number of agents, $m_{j}\in\mathbb{Z}_{+}$. In other words,
\begin{align*}
  \sum_{i}M_{ij} & \leq c_{j} & \forall j\in O &  & \tag{Cap} \\
  \sum_{i}M_{ij} & \geq m_{j} & \forall j\in O &  & \tag{Min}
\end{align*}
As shorthand, we refer to $O_{m}:=\{j:m_{j}>0\}$ as ``minimum objects.''
A deterministic allocation $M$ is \textbf{allowable} if it satisfies
the demand, capacity, and minimum conditions shown above. Given a
problem, the set of allowable deterministic allocations is $D(N,O)$,
or simply $D$. When we refer to the set of deterministic allocations,
we mean $D$. Without loss of generality, we impose that $c_{j}\leq N$
for all $j$. We also assume the problem is feasible; i.e. $\sum_{j}c_{j}\geq Nd$
and $\sum_{j}m_{j}\leq N$.

A matrix $\mu\in[0,1]^{N\times O}$ is a \textbf{random allocation},
where $\mu_{ij}$ represents agent $i$'s probability of being assigned
to object $j$. Row $\mu_{i}$ represents agent $i$'s lottery over
the objects. A random allocation is \textbf{implementable} if it is
a convex combination of allowable deterministic allocations. That
is, let $\Delta D:=\text{conv}(D)$; $\mu$ is implementable if $\mu\in\Delta D$.

A feature of our model is that many constraints on allocations can
be written in terms of marginal sums over the agents. As shorthand,
we denote $M_{j}=\sum_{i}M_{ij}$ and $\mu_{j}=\sum_{i}\mu_{ij}$.
Note that we can rewrite our condition for capacity constraints as
$M_{j}\leq c_{j}$ and our condition for minimum constraints as $M_{j}\geq m_{j}$.
We will make use of this notation throughout the paper.

It is straight forward but nontrivial that implementable random allocations
are characterized by the same inequalities.
\begin{prop}
  \label{prop:deltaD}Given a market $(N,O,\succ)$, the set $\Delta D$
  of implementable random allocations is characterized by the following
  system :
  \begin{align*}
    \sum_{j}\mu_{ij} & =d         &  & \forall i\in N     \\
    \mu_{j}          & \leq c_{j} &  & \forall j\in O     \\
    \mu_{j}          & \geq m_{j} &  & \forall j\in O_{m}
  \end{align*}
  where $\mu\in[0,1]^{N\times O}$.
\end{prop}
\begin{proof}
  Appendix \ref{sec:prelimproofs}.
\end{proof}
Each agent $i\in N$ has a strict preference ranking $\succ_{i}$
over the objects. A \textbf{mechanism} is $f:(N,O,\succ)\rightarrow\Delta D$.
We typically suppress the first inputs and write $f(\succ)$. In words,
a mechanism takes in the model primitives and agent preferences over
objects and reports an implementable random allocation.

A random allocation is \textbf{ex post efficient} if it is a convex
combination of Pareto efficient deterministic allocations. This is
not difficult to achieve -- random serial dictatorship (sometimes
called random priority) yields ex post efficient allocations. Using
first order stochastic dominance (FOSD), we also endow agents with
a partial preference ordering $\geq_{i}^{FOSD}$ over random allocations.
Without loss of generality, label the objects such that $o_{1}\succ_{i}o_{2}\succ_{i}\cdots\succ_{i}o_{O}$.
Given random allocations $\mu$ and $\nu$, denote $\mu_{i}\geq_{i}^{FOSD}\nu_{i}$
if
\[
  \sum_{j=1}^{k}\mu_{ij}\geq\sum_{j=1}^{k}\nu_{ij}\qquad\forall k=1,\ldots,O.
\]
Additionally, $\mu_{i}>_{i}^{FOSD}\nu_{i}$ if $\mu_{i}\geq_{i}^{FOSD}\nu_{i}$
and $\mu_{i}\neq\nu_{i}$. A random allocation $\mu$ is \textbf{SD
  efficient} if there does not exist $\nu$ such that $\nu_{i}\geq_{i}^{FOSD}\mu_{i}$
for all $i\in N$ and $\nu_{i}>_{i}^{FOSD}\mu_{i}$ for at least one
$i$. A mechanism is SD efficient if it always selects SD efficient
random allocations.\footnote{SD efficiency is called \emph{ordinal efficiency} in \textcite{BM01}.}

Random serial dictatorship is ex post efficient but not SD efficient.
We reproduce an example from \textcite{BM01} below to illustrate.
\begin{example}
  Let $N=\{1,2,3,4\}$ and $O=\{o_{1},o_{2},o_{3},o_{4}\}$. Let preferences
  be given by:
  \begin{align*}
    1,2 & :o_{1}\succ o_{2}\succ o_{3}\succ o_{4} \\
    3,4 & :o_{2}\succ o_{1}\succ o_{4}\succ o_{3}
  \end{align*}
  Random serial dictatorship (with uniform probability on order of selection)
  results in the random allocation $\nu$, shown below:
  \[
    \begin{array}{c|cccc}
      \nu      & o_{1} & o_{2} & o_{3} & o_{4} \\
      \hline 1 & 5/12  & 1/12  & 5/12  & 1/12  \\
      2        & 5/12  & 1/12  & 5/12  & 1/12  \\
      3        & 1/12  & 5/12  & 1/12  & 5/12  \\
      4        & 1/12  & 5/12  & 1/12  & 5/12
    \end{array}
  \]
  For example, agent 1 receives $o_{1}$ if he is first priority or
  second priority \emph{after 3 or 4}, which occurs with probability
  $5/12$. However, each agent's lottery under $\nu$ is stochastically
  dominated by his lottery under the implementable random allocation
  $\mu$ shown below:
  \[
    \begin{array}{c|cccc}
      \mu      & o_{1} & o_{2} & o_{3} & o_{4} \\
      \hline 1 & 1/2   & 0     & 1/2   & 0     \\
      2        & 1/2   & 0     & 1/2   & 0     \\
      3        & 0     & 1/2   & 0     & 1/2   \\
      4        & 0     & 1/2   & 0     & 1/2
    \end{array}
  \]
  Here, each agent's probability under $\nu$ on his second (fourth)
  favorite is moved to his first (third) favorite.
\end{example}
SD efficiency is an ordinal notion of efficiency for random allocations.
As \textcite{BM01} note, its definition does not assume agents have
von Neumann-Morgenstern (vNM) expected utilities. Therefore, from
the perspective of the market designer, SD efficiency is attractive
because it does not require eliciting vNM utilities. It is straightforward
to see that if an allocation is SD efficient, it is Pareto efficient
according to \emph{some} profile of vNM utilities.\footnote{This is \textcite{BM01} Lemma 2.}

We also consider two notions of fairness. A mechanism is \textbf{anonymous}
if re-indexing the agents simply re-indexes their allocations in the
same way. Formally, if $\pi$ is a permutation, let $\succ^{\pi}$
be the permutation of $\succ$ and $\mu^{\pi}$ be the same permutation
of a random allocation. The mechanism $f$ is anonymous if $f(\succ^{\pi})=f(\succ)^{\pi}$.
A random allocation is \textbf{envy free} if for all $i,i'\in N$,
we have $f_{i}(R)\geq_{i}^{FOSD}f_{i'}(R)$. That is, $i$ does not
prefer another agent's lottery. A mechanism is envy free if it only
selects envy free random allocations.

Finally, we consider incentives. A mechanism $f$ is \textbf{weak
  strategyproof} if, for any (true) preferences $\succ$ and any $i\in N$
with misreport $\succ_{i}'$,
\[
  f_{i}(\succ_{i}',\succ_{-i})\geq_{i}^{FOSD}f_{i}(\succ_{i},\succ_{-i})\Rightarrow f_{i}(\succ_{i}',\succ_{-i})=f_{i}(\succ_{i},\succ_{-i}).
\]
In words, if $i$ unilaterally misreports his preferences, it does
not give him a first order stochastically dominant lottery. Note that
the misreport may result in a FOSD noncomparable allocation.

We will present a greedy mechanism in the style of Probabilistic Serial.
However, it is not obvious from the primitive constraints how to implement
this greedy procedure while ensuring that object minimums are met.
The primary work of this paper is to characterize the paths from the
origin to implementable random allocations; the greedy procedure then
chooses a path in this set.

\section{Paths to implementable random allocations}

The first step is to characterize implementable random allocations
$\Delta D$ from our model primitives. An intuitive way to do this
is to enumerate $D$, then compute the convex hull of $D$. However,
this is impractical as as Example \ref{exa:factorial} demonstrates.
We will therefore give an analytic description of $\Delta D$ that
avoids enumerating the set of deterministic allocations.

We next consider the \textbf{lower contour set} of $\Delta D$ introduced
by B22, defined as
\[
  \text{lcs}(\Delta D):=\left\{ \mu'\in\mathbb{R}_{+}^{N\times O}:\exists\mu\in\Delta D\text{ where }\mu'\leq\mu\right\} .
\]
Intuitively, $\text{lcs}(\Delta D)$ are sub-random allocations that
can be completed to implementable random allocations. This characterizes
valid paths for an eating algorithm to assign probability to agents.

The key insight of B22 is that for any general constraints (including
those beyond the present paper), $\text{lcs}(\Delta D)$ can be expressed
as the intersection of the half-spaces defined by a set of positive
linear inequalities.
\begin{prop}[\textcite{Balbuzanov22}]
  \label{prop:Balb1} For any $D$, there exists an essentially unique
  $A\geq0$ and $b\geq0$ such that
  \[
    \text{lcs}(\Delta D)=\{\mu\in\mathbb{R}_{+}^{N\times O}:A\mu\le b\}.
  \]
\end{prop}
This also implies $\text{lcs}(\Delta D)$ is a convex polytope in
the positive orthant. These inequalities lead naturally to a greedy
algorithm à la PS. An agent consumes his most preferred object until
one of these inequalities binds, and then must move on to his next
favored object. Since all inequalities are positive, an inequality
that begins to bind will bind forever.

While the above result gives a theoretical guarantee for the existence
of these positive linear inequalities, they must still be computed
in practice. The points in $D$ along with their projections onto
the axes constitute the vertices of the convex polytope $\text{lcs}(\Delta D)$.
Finding the inequalities defining the facets of the polytope from
the vertices is the \emph{facet enumeration problem (FEP)}. This problem
is known to be NP-hard and has complexity scaling in the number of
vertices and dimensions, which both grow quickly by introducing new
objects or agents. We refer readers to \textcite{Seidel18} for a reference
text.

We adapt this machinery to our setting. We start with a list of constraints
$(m_{j},c_{j})_{j\in O}$ as our primitive, not the set of allowable
deterministic allocations, $D$, since size of $D$ (and complexity
of enumerating it) can be factorial. This factorial size $D$ would
then be fed into the exponential complexity FEP. Thus even in moderately
sized problems, this approach is impractical. Our approach avoids
this issue -- we directly find the analytic description of $\text{lcs}(\Delta D)$.

Once we have expressed $\text{lcs}(\Delta D)$ using a system of positive
linear inequalities, we can immediately apply GCPS. For convenience,
we give the definition of GCPS below.
\begin{defn}[\textcite{Balbuzanov22}]
  \textbf{ Generalized Constrained Probabilistic Serial (GCPS) mechanism}.
  Time runs continuously from $t\in[0,1]$; each $t$ is associated
  with a sub-random allocation $\mu^{t}$. Given a preference profile
  $\succ$ and constraints $A,b$ that define $\text{lcs}(\Delta D)$,
  object $j$ is available to agent $i$ and time $t$ if none of the
  inequalities involving $\mu_{ij}$ bind at that time. At time $t$,
  each agent claims with uniform speed remaining probability shares
  of his favorite reported object $j$ among those available to him
  at $t$. The procedure ends at time $t=1$ with output $\mu^{1}$,
  which is an implementable random allocation.
\end{defn}
However, even after finding $\text{lcs}(\Delta D)$, one more barrier
exists with regards to practical use: GCPS defines a continuous-time
eating process, not an algorithm in the strict sense. In the $d=1$
case we also express our mechanism as an algorithm (i.e., a finite
sequence of discrete steps).

\section{Minimums probabilistic serial mechanism}

We now describe our mechanism and its properties.

As before, we have a set of agents $N$ and objects $O$. Each object
$j$ has a capacity $c_{j}>0$ and a minimum $m_{j}\geq0$. Our goal
is to find the system of positive linear inequalities referenced by
Proposition \ref{prop:Balb1}, $A\mu\leq b$, which define $\text{lcs}(\Delta D)$.
We first present the leading case of unit demand, $d=1$, which has
the simplest exposition and strongest results.

\begin{prop}
  \label{prop:rand_mins}Given a market $(N,O,\succ)$ and $d=1$, the
  set $\text{lcs}(\Delta D)$ of allowable sub-random allocations $\mu\in[0,1]^{N\times O}$
  is characterized by the following system:
  \begin{align}
    \sum_{j}\mu_{ij}                 & \leq1                     &  & \forall i\in N           & \tag{Unit demand - relaxed}\nonumber \\
    \mu_{j}                          & \leq c_{j}                &  & \forall j\in O           & \tag{Cap}\nonumber                   \\
    \sum_{j\in O\backslash S}\mu_{j} & \leq N-\sum_{j\in S}m_{j} &  & \forall S\subseteq O_{m} & \tag{Min-III }\label{eq:Min-III}
  \end{align}
\end{prop}
\begin{proof}
  Appendix \ref{sec:prelimproofs}.
\end{proof}
\begin{rem}
  The third group of inequalities is for any set $S\subseteq O_{m}$.
  Since $m_{j}=0$ for all $j\notin O_{m}$, it is equivalent to $\forall S\subseteq O$.
\end{rem}
Proposition \ref{prop:rand_mins} provides an analytic description
of the inequalities which are guaranteed to exist by Proposition \ref{prop:Balb1}
for our setting. We can now apply GCPS to express our mechanism as
a continuous-time eating process.
\begin{defn}
  \textbf{Minimums Probabilistic Serial (MPS) mechanism}. Fix $(N,O,\succ)$
  and the profile of constraints $(m_{j},c_{j})_{j\in O}$. Time runs
  continuously from $t\in[0,1]$; each $t$ is associated with a sub-random
  allocation $\mu^{t}$. Object $j$ is available at time $t$ if none
  of the inequalities in Proposition \ref{prop:rand_mins} involving
  $\mu_{j}$ bind at that time. At time $t$, each agent claims with
  uniform speed probability shares of his favorite reported object $j$
  among those available at $t$. The procedure ends at time $t=1$ with
  output $\mu^{1}$, which is an implementable random allocation.
\end{defn}
We therefore describe MPS without the need to enumerate $D$ from
the primitive constraints nor to solve the facet enumeration problem,
both of which are non-polynomial time procedures. An algorithmic definition
is contained later in this section.
\begin{thm}
  \label{thm:MPS}The MPS mechanism with unit demand is SD efficient,
  anonymous, envy free, and weak strategyproof.
\end{thm}
\begin{proof}
  SD efficiency and anonymity are inherited from B22. The remainder
  is in Appendix \ref{proof:envyfree} and Appendix \ref{proof:wsp}.
\end{proof}
As Theorem 1 shows, MPS maintains the most desirable properties of
PS from \textcite{BM01}. In particular, weak strategyproofness is lost
in the more general setting of B22 but preserved in our setting. The
intuition for SD efficiency matches that in \textcite{BM01} and B22.
Agents consume their favored objects until they become unavailable.
Since the linear inequalities governing availability are positive,
an object becomes unavailable forever. Thus there is no possibility
that an object closes then reopens later. Envy freeness has a similar
intuition. Since objects are available to everyone or to no one, an
agent can never envy another agent's allocation. Anonymity is also
immediate for this reason.

Weak strategyproofness is much more challenging to prove. As becomes
clear in the description of Algorithm \ref{alg:MPS}, the total amount
of an object that is available to consume may increase if agents'
consumption of $O_{m}$ is manipulated. This possibility is not present
in \textcite{BM01} as the only constraints are maximum capacities. For
example, it is possible that in a truthful run of MPS, the total consumption
of object $a$ is capped by $m_{a}$, but under a misreport run, total
consumption of $a$ might be able to exceed $m_{a}$ if agents' consumption
is shifted into $O_{m}$. The proof (contained in Appendix \ref{proof:wsp})
shows that while this is possible, it can never benefit the misreporter.
\begin{rem}
  MPS is not an application of the Generalized Probabilistic Serial
  (GPS) mechanism contained in BCKM13. Given a bihierarchical set $\mathcal{H}$
  of constraint sets $S\subseteq N\times O$, GPS allows for constraints
  of the form
  \[
    \underline{q}_{S}\leq\sum_{(i,j)\in S}\mu_{ij}\leq\bar{q}_{S}
  \]
  where $\underline{q}_{S}=0$ for all $S\in\mathcal{H}$ that are not
  rows. In general, it is impossible to represent the constraints in
  our model this way. Details are in Appendix \ref{sec:BCKM}.
\end{rem}
The conditions in Proposition \ref{prop:rand_mins} can be re-written
more compactly. The advantage of the following representation is that
it leads easily to an algorithm that runs in a finite number of discrete
steps.
\begin{prop}
  \label{prop:rand_mins_equiv} Given a market $(N,O,\succ)$, $\text{lcs}(\Delta D)$
  is the set of all $\mu\in[0,1]^{N\times O}$ satisfying the following
  conditions:
  \begin{align}
    \sum_{j}\mu_{ij}              & \leq1      & \forall i\in N &  & \tag{Unit demand - relaxed}\nonumber \\
    \mu_{j}                       & \leq c_{j} & \forall j\in O &  & \tag{Cap}\label{eq:Cap-IV}           \\
    \sum_{j}\max\{m_{j},\mu_{j}\} & \leq N     &                &  & \tag{Min-IV}\label{eq:Min-IV}
  \end{align}
\end{prop}
\begin{proof}
  Appendix \ref{sec:prelimproofs}.
\end{proof}
While \textcite{BM01} contains an algorithm for the PS mechanism, B22
does not for the GCPS mechanism.\footnote{Of course, his setting is significantly more complicated.}
We take advantage of the additional structure in our model to express
MPS as an algorithm.

The characterization of $\text{lcs}(\Delta D)$ in Proposition \ref{prop:rand_mins_equiv}
suggests three things that can occur to end a step of the MPS algorithm.
If the first inequality binds for any agent, it means that this agent
has consumed his total share of probability and must stop eating.
Since we restrict to uniform eating speed and unit demand, this occurs
for all agents at time $t=1$. When the second inequality binds for
any object, it means that this object has reached its maximum capacity
and can no longer be consumed. Lastly, when the third inequality binds,
the only way for agents to continue eating is to consume objects for
which $\mu_{j}<m_{j}$. That is, agents must consume from objects
whose minimum requirements are not yet satisfied.

The algorithm will run in discrete steps, denoted $s$. For any $S\subseteq O$
and $j\in S$, let $N(j,S)=\{i\in N:j\succsim_{i}k\quad\forall k\in S\}$.
That is, $N(j,S)$ is the set of agents who would consume $j$ among
the objects in $S$. Also, let $n(j,S)=|N(j,S)|$. The objective is
to find the time $t^{s}$ at which the current step $s$ will end.
Let $O^{s-1}$ be objects available at step $s$, and $O_{m}^{s-1}\subseteq O^{s-1}$
be those that have not yet met their minimum requirements.

For any object $j\in O^{s-1}$, define
\[
  t_{c}^{s}(j)=\begin{cases}
    t^{s-1}+\frac{c_{j}-\mu_{j}^{s-1}}{n(j,O^{s-1})} & \text{if }n(j,O^{s-1})>0 \\
    +\infty                                          & \text{otherwise.}
  \end{cases}
\]
Note that $t_{c}^{s}(j)$ expresses the time at which object $j$
would reach its capacity if all agents who prefer $j$ among $O^{s-1}$
begin eating from $j$ at time $t^{s-1}$. By construction, we know
that $\mu_{j}^{s-1}<c_{j}$ for all $j\in O^{s-1}$.

For any object $j\in O_{m}^{s-1}$, define
\[
  t_{m}^{s}(j)=\begin{cases}
    t^{s-1}+\frac{m_{j}-\mu_{j}^{s-1}}{n(j,O^{s-1})} & \text{if }n(j,O^{s-1})>0 \\
    +\infty                                          & \text{otherwise.}
  \end{cases}
\]
Note that $t_{m}^{s}(j)$ expresses the time at which object $j$
would satisfy its minimum requirement if all agents who prefer $j$
among $O^{s-1}$ begin eating from $j$ at time $t^{s-1}$. By construction,
we know that $\mu_{j}^{s-1}<m_{j}$ for all $j\in O_{m}^{s-1}$.

Last, define
\[
  t_{O}^{s}=\begin{cases}
    t^{s-1}+\frac{N-\sum_{j\in O_{m}^{s-1}}m_{j}-\sum_{j\notin O_{m}^{s-1}}\mu_{j}}{\sum_{j\notin O_{m}^{s-1}}n(j,O^{s-1})} & \text{if }n(j,O^{s-1})>0\text{ for some }j\notin O_{m}^{s-1} \\
    +\infty                                                                                                                 & \text{otherwise.}
  \end{cases}
\]
Note that $t_{O}^{s}$ expresses the time at which the minimum condition
from Proposition \ref{prop:rand_mins_equiv} binds, if this occurs
at step $s$.

\begin{algorithm}[H]
  \caption{Minimums Probabilistic Serial (MPS)}
  \label{alg:MPS} \begin{algorithmic}[1] \State Initialize $t^{0}=0$,
    $O_{m}^{0}=O_{m}$, $O^{0}=O$, $\mu^{0}=0$, $s=0$, and $F=\text{False}$.
    \While{$t^{s}<1$} \State Set $s=s+1$. \State Set $t^{s}=\min\big\{\{t_{c}^{s}(j):j\in O^{s-1}\setminus O_{m}^{s-1}\}\cup\{t_{m}^{s}(j):j\in O_{m}^{s-1}\}\cup\{t_{O}^{s},1\}\big\}$.
    \If{$t^{s}=t_{O}^{s}$} \State Set $F=\text{True}$. \EndIf \State
    Set $O_{m}^{s}=O_{m}^{s-1}\setminus\{j\in O_{m}^{s-1}:t_{m}^{s}(j)=t^{s}\}$.
    \If{$F$} \State Set $O^{s}=O_{m}^{s}$. \Else \State Set $O^{s}=O^{s-1}\setminus\{j\in O^{s-1}\setminus O_{m}^{s-1}:t_{c}^{s}(j)=t^{s}\}$.
    \EndIf \For{$j\in O$} \For{$i\in N$} \State Set
    \[
      \mu_{ij}^{s}=\begin{cases}
        \mu_{ij}^{s-1}+(t^{s}-t^{s-1}) & \text{if }j\in O^{s-1}\text{ and }i\in N(j,O^{s-1}), \\
        \mu_{ij}^{s-1}                 & \text{otherwise.}
      \end{cases}
    \]
    \EndFor \EndFor \EndWhile \State \Return $[\mu_{ij}^{s}]$. \end{algorithmic}
\end{algorithm}

Note that $t^{s}=t_{O}^{s}$ at most once during the run of the algorithm,
since $O^{s}=O_{m}^{s}$ and therefore $t_{O}^{s}=+\infty$ for every
subsequent step $s$ after the step at which $t^{s}=t_{O}^{s}$. We
will denote this time where $t^{s}=t_{O}^{s}$ by $\tau$; the role
of $\tau$ in our mechanism is important in proving weak strategyproofness
of MPS. By misreporting their preferences, agents may be able to alter
$\tau$. We show, however, that they can never benefit from doing
so.

Intuitively, the algorithm works as follows. At time $t=0$, agents
begin eating from their favorite objects. If an object reaches its
capacity, it is closed to all agents and agents begin eating from
their next favorite remaining object. At some time -- which corresponds
to $\tau$ as we have defined it above -- it may be necessary to
close all objects once they have met their minimum requirements in
order to ensure that all objects' minimum requirements are met. That
is, at time $\tau$, any object that has already met its minimum requirement
is closed, even if it is not yet at its capacity. Thereafter, objects
are closed upon reaching their minimums.

Thus we are able to use our results to describe MPS as an algorithm
that runs in a finite number of discrete steps. We characterize this
mechanism without enumerating the deterministic allocations nor solving
the facet enumeration problem, which would be a non-polynomial complexity
problem with factorial inputs. Each step of our algorithm requires
calculating at most $|O|+1$ values. At the end of each step, either
1) an object reaches its capacity; 2) an object's minimum is satisfied;
3) the minimums become binding across the market; or 4) $t=1$. Therefore
there are at most $2|O|+2$ steps.

There does not exist an algorithmic description of GCPS. Of course,
B22's setting is much more complicated. However, in our setting, we
overcome a significant hurdle to practical usage by providing a polynomial
time algorithm.

We present a simple example to illustrate MPS.
\begin{example}
  Let $N=\{1,2,3\}$ and $O=\{o_{1},o_{2},o_{3}\}$ with the following
  properties:
  \[
    \begin{array}{c|cc}
                   & c_{j} & m_{j} \\
      \hline o_{1} & 2     & 1     \\
      o_{2}        & 2     & 1     \\
      o_{3}        & 2     & 0
    \end{array}
  \]
  Let preferences be given by:
  \begin{align*}
    1,2 & :o_{1}\succ_{i}o_{2}\succ_{i}o_{3} \\
    3   & :o_{3}\succ_{i}o_{1}\succ_{i}o_{2}
  \end{align*}

  For convenience, we denote deterministic allocations as $(M_{1},M_{2},M_{3})$,
  the number of agents assigned to each object; and analogously for
  random allocations. The allowable discrete allocations are $D=\{(2,1,0),(1,2,0),(1,1,1)\}$.
  Following Proposition \ref{prop:rand_mins}, $\text{lcs}(\Delta D)$
  is characterized by
  \begin{align*}
    \mu_{i}         & \leq1      &  & \forall i\in N \\
    \mu_{j}         & \leq c_{j} &  & \forall j\in O \\
    \mu_{3}         & \leq1                          \\
    \mu_{1}+\mu_{3} & \leq2                          \\
    \mu_{2}+\mu_{3} & \leq2.
  \end{align*}

  To implement MPS, agents consume their favorite available object until
  an inequality binds. Concretely, the steps are
  \begin{enumerate}
    \item Agents 1 and 2 consume $o_{1}$. Agent 3 consumes $o_{3}$. At $t=2/3$,
          $\mu_{1}+\mu_{3}\leq2$ binds. Afterwards, only $o_{3}$ is available
          to consume.
    \item Agents 1, 2, and 3 consume $o_{3}$ until $t=1$.
  \end{enumerate}
  The final random allocation produced is
  \[
    \begin{array}{c|ccc}
      \mu      & o_{1} & o_{2} & o_{3} \\
      \hline 1 & 2/3   & 1/3   & 0     \\
      2        & 2/3   & 1/3   & 0     \\
      3        & 0     & 1/3   & 2/3
    \end{array}
  \]
\end{example}

\subsection{General demand $d$}

We now present results for general demand $d$. As before, the main
task is to characterize $\text{lcs}(\Delta D)$.
\begin{prop}
  \label{prop:general_d} Given a market $(N,O,\succ)$ and $d\in\mathbb{N}$,
  the set $\text{lcs}(\Delta D)$ of allowable sub-random allocations
  $\mu\in[0,1]^{N\times O}$ is characterized by the following system:
  \begin{align}
    \sum_{i\in N\backslash T,j\in S}\mu_{ij} & \leq|T|\ (|O\backslash S|-d)+\sum_{j\in S}c_{j}    &  & \forall S\subseteq O,\forall T\subseteq N     & \tag{Cap-V }\label{eq:Cap-V} \\
    \sum_{i\in T,j\in O\backslash S}\mu_{ij} & \leq|T|\ d+|N\backslash T|\ |S|-\sum_{j\in S}m_{j} &  & \forall S\subseteq O_{m},\forall T\subseteq N & \tag{Min-V }\label{eq:Min-V}
  \end{align}
\end{prop}
\begin{proof}
  Appendix \ref{sec:prelimproofs}.
\end{proof}
\begin{rem}
  The demand restriction $\mu_{i}\leq d$ is contained in \ref{eq:Min-V}.
  Let $S=\emptyset$ and $T=\{i\}$.
\end{rem}
The definition of the MPS mechanism is analogous; agent $i$ consumes
his favorite object $j$ until an inequality binds, then consumes
his favorite remaining object. The formal description is omitted.
As is evident in the number of inequalities, general demand greatly
increases the complexity of the problem. We leave an algorithmic description,
if one exists, as future research.

MPS preserves our efficiency and fairness properties in this general
case. However, as noted in \textcite{K09}, weak strategyproofness does
not hold.
\begin{thm}
  \label{thm:MPS_general}The MPS mechanism is SD efficient, anonymous,
  and envy free.
\end{thm}
\begin{proof}
  Appendix \ref{proof:envyfree} and Appendix \ref{proof:wsp}.
\end{proof}

\subsection{Implementation}

Once we arrive at an implementable random allocation $\mu$ via MPS,
it remains to implement $\mu$ by finding a randomization over $D$.
Since $\mu$ may not be bistochastic, the Birkhoff-von Neumann theorem
is not applicable. Recall that MPS is not a special case of BCKM13's
Generalized Probabilistic Serial (GPS) mechanism, since GPS does not
accept minimums on the constraint sets. However, we can still apply
the implementation results in BCKM13. Thus, after using MPS to arrive
at an implementable random allocation, we are able to apply results
from BCKM13 to randomize over allowable deterministic allocations.
We refer the reader to their paper for details.

\section{Conclusion}

This paper deals with the allocation of objects to agents in a unit
demand setting. It is not difficult to motivate minimums in this setting.
Minimums might reflect lower bounds on the sizes of viable classes,
projects, or basketball teams. However, finding implementable allocations
that respect the minimums is nontrivial. We construct a new mechanism,
the Minimums Probabilistic Serial (MPS) mechanism, that yields an
implementable random allocation when objects may have both minimum
and maximum restrictions. We show that MPS is SD efficient, anonymous,
and envy free. In the unit demand case, it is also weak strategyproof.
The last property is preserved, unlike in the more general setting
of \textcite{Balbuzanov22}. Importantly, MPS with unit demand is also
computationally feasible -- we provide a polynomial time algorithmic
description in this case. We hope that our results help to bridge
the gap between the literature and actual implementation in real applications.

\newpage{}

\printbibliography[title=References]

\newpage{}

\appendix

\section{Proofs}

\subsection{Proofs of Propositions \ref{prop:deltaD}, \ref{prop:rand_mins},
  and \ref{prop:rand_mins_equiv}}

\label{sec:prelimproofs}
\begin{proof}[\textbf{Proof of Proposition \ref{prop:deltaD}.}]
  Denote
  \[
    D=\left\{ M\in\{0,1\}^{N\times O}:M_{i}=d,\quad m_{j}\leq M_{j}\leq c_{j}\right\}
  \]
  and
  \[
    R=\left\{ \mu\in[0,1]^{N\times O}:\mu_{i}=d,\quad m_{j}\leq\mu_{j}\leq c_{j}\right\}
  \]
  We wish to show that $\Delta D=R$. That $\Delta D\subseteq R$ is
  immediate: $D\subset R$, and $R$ is convex. We show $\Delta D\supseteq R$.
  Write $R$ as $R=\{\mu:\Gamma\mu\leq\kappa\}$, where $\Gamma$ and
  $\kappa$ are not necessarily positive. By \textcite{Edmonds70}\footnote{Line 39; cf. BCKM13},
  since $\Gamma$ represents the incidence matrix of a bihierarchy,
  $\Gamma$ is totally unimodular (TUM). Then the extreme points of
  $R$ are integral. It is straight forward that $R\cap\mathbb{Z}^{N\times O}\subseteq D$.
  Thus $\Delta D\supseteq R$ as desired.
\end{proof}
\medskip{}

We now turn to the proof of Proposition \ref{prop:general_d}, of
which Proposition \ref{prop:rand_mins} is a special case. We first
note a useful graph theoretic result. Let $G=(V,E)$ be a directed
network, with finite nodes $V$ and directed arcs $E\subseteq V\times V$.
A \textbf{circulation} is an assignment of flow values $x_{ij}$ for
$(i,j)\in E$ such that $\sum_{(i,j)\in E}x_{ij}-\sum_{(j,i)\in E}x_{ji}=0$
for all $i\in V$. Further, let there be bounds $l_{ij}\leq x_{ij}\leq u_{ij}$
for all $(i,j)\in E$. A circulation is feasible if it respects these
bounds.
\begin{thm*}[Hoffman's existence theorem, \citeyear{Hoffman60}.]
  A feasible circulation exists if and only if, for every set $X\subseteq V$,
  \[
    \sum_{(i,j)\in E,i\in S,j\notin S}l_{ij}\leq\sum_{(i,j)\in E,i\in S,j\notin S}u_{ij}.
  \]
  That is, the lower bound of the in-arcs is weakly less than the upper
  bound of the out-arcs.
\end{thm*}
We now prove Proposition \ref{prop:general_d}.
\begin{proof}[\textbf{Proof of Proposition \ref{prop:general_d}.}]
  Fix some $\nu\in[0,1]^{N\times O}$ as a candidate (sub-)random
  allocation. It suffices to show that a random allocation $\mu\geq\nu$
  such that
  \begin{align*}
     & \nu_{ij}\leq\mu_{ij}\leq1  \\
     & \mu_{i}=d                  \\
     & m_{j}\leq\mu_{j}\leq c_{j}
  \end{align*}
  exists if and only if the inequalities in the proposition are respected.
  This is equivalent to a feasible circulation problem. Construct a
  digraph with vertices $V=\{s,t\}\cup N\cup O$; and edges $E$
\end{proof}
\begin{itemize}
  \item $s\rightarrow i$ for all $i\in N$, with $l_{si}=u_{si}=d$
  \item $i\rightarrow j$ for all $i\in N,j\in O$, with $l_{ij}=\nu_{ij}$
        and $u_{ij}=1$
  \item $j\rightarrow t$ for all $j\in O$, with $l_{jt}=m_{j}$ and $u_{jt}=c_{j}$;
        and
  \item $t\rightarrow s$, with $l_{ts}=u_{ts}=Nd$.
\end{itemize}
We apply Hoffman's existence theorem and check all subsets of vertices
$X$; consider any subsets $T=X\cap N$ and $S=X\cap O$, with four
cases depending on inclusion of $s$or $t$ in $X$.
\begin{enumerate}
  \item \textbf{Case 1.} $s\in X,t\notin X$. The upper bound of the out-arcs
        is
        \[
          \underset{\text{s\ensuremath{\rightarrow}N\textbackslash T}}{\underbrace{d\ |N\backslash T|}}+\underset{\text{T\ensuremath{\rightarrow}O\textbackslash S}}{\underbrace{|T||O\backslash S|}}+\underset{S\rightarrow t}{\underbrace{\sum_{j\in S}c_{j}}}.
        \]
        The lower bound of the in-arcs is
        \[
          \underset{\text{N\textbackslash T\ensuremath{\rightarrow}S}}{\underbrace{\sum_{i\in N\backslash T,j\in S}\nu_{ij}}}+\underset{t\rightarrow s}{\underbrace{Nd}}
        \]
        The associated if and only if condition is
        \[
          \sum_{i\in N\backslash T,j\in S}\nu_{ij}+Nd\leq d\ |N\backslash T|+|T|\ |O\backslash S|+\sum_{j\in S}c_{j}
        \]
        \[
          \iff\sum_{i\in N\backslash T,j\in S}\nu_{ij}\leq|T|\ (|O\backslash S|-d)+\sum_{j\in S}c_{j}.
        \]
  \item \textbf{Case 2.} $s\notin X,t\notin X$. The upper bound of the out-arcs
        is
        \[
          \underset{\text{T\ensuremath{\rightarrow}O\textbackslash S}}{\underbrace{|T|\ |O\backslash S|}}+\underset{\text{S\ensuremath{\rightarrow}t}}{\underbrace{\sum_{j\in S}c_{j}}}.
        \]
        The lower bound of the in-arcs is
        \[
          \underset{s\rightarrow T}{\underbrace{|T|\ d}}+\underset{N\backslash T\rightarrow S}{\underbrace{\sum_{i\in N\backslash T,j\in S}\nu_{ij}}}.
        \]
        The associated if and only if condition is the same as in Case 1.
  \item \textbf{Case 3.} $s\notin X,t\in X$. The upper bound of the out-arcs
        is
        \[
          \underset{\text{T\ensuremath{\rightarrow}O\textbackslash S}}{\underbrace{|T|\ |O\backslash S|}}+\underset{\text{t\ensuremath{\rightarrow}s}}{\underbrace{Nd}}.
        \]
        The lower bound of the in-arcs is
        \[
          \underset{s\rightarrow T}{\underbrace{|T|\ d}}+\underset{N\backslash T\rightarrow S}{\underbrace{\sum_{i\in N\backslash T,j\in S}\nu_{ij}}}+\underset{O\backslash S\rightarrow t}{\underbrace{\sum_{j\in O\backslash S}m_{j}}}.
        \]
        The associated if and only if condition is
        \[
          |T|\ d+\sum_{i\in N\backslash T,j\in S}\nu_{ij}+\sum_{j\in O\backslash S}m_{j}\leq|T|\ |O\backslash S|+Nd
        \]
        \[
          \iff\sum_{i\in N\backslash T,j\in S}\nu_{ij}\leq|N\backslash T|\ d+|T|\ |O\backslash S|-\sum_{j\in O\backslash S}m_{j}.
        \]
        Swapping the roles of $N\backslash T\leftrightarrow T$ and $O\backslash S\leftrightarrow S$
        gives
        \[
          \sum_{i\in T,j\in O\backslash S}\nu_{ij}\leq|T|\ d+|N\backslash T|\ |S|-\sum_{j\in S}m_{j}.
        \]
  \item \textbf{Case 4. }$s\in X,t\in X$. The upper bound of the out-arcs
        is
        \[
          \underset{\text{s\ensuremath{\rightarrow}N\textbackslash T}}{\underbrace{d\ |N\backslash T|}}+\underset{\text{T\ensuremath{\rightarrow}O\textbackslash S}}{\underbrace{|T||O\backslash S|}}.
        \]
        The lower bound of the in-arcs is
        \[
          \underset{N\backslash T\rightarrow S}{\underbrace{\sum_{i\in N\backslash T,j\in S}\nu_{ij}}}+\underset{O\backslash S\rightarrow t}{\underbrace{\sum_{j\in O\backslash S}m_{j}}}.
        \]
        The associated if and only if condition is the same as in Case 3.
\end{enumerate}
Thus $\nu\in\text{lcs}(\Delta D)$ if and only if
\begin{align*}
  \sum_{i\in N\backslash T,j\in S}\nu_{ij} & \leq|T|\ (|O\backslash S|-d)+\sum_{j\in S}c_{j}    \\
  \sum_{i\in T,j\in O\backslash S}\nu_{ij} & \leq|T|\ d+|N\backslash T|\ |S|-\sum_{j\in S}m_{j}
\end{align*}
as desired.

\medskip{}
We prove the unit demand inequalities in Proposition \ref{prop:rand_mins}
as a special case of general demand.
\begin{proof}[\textbf{Proof of Proposition \ref{prop:rand_mins}.}]
  Let $d=1$; then the inequalities become
  \begin{align*}
    \sum_{i\in N\backslash T,j\in S}\mu_{ij} & \leq|T|\ (|O\backslash S|-1)+\sum_{j\in S}c_{j} \\
    \sum_{i\in T,j\in O\backslash S}\mu_{ij} & \leq|T|+|N\backslash T|\ |S|-\sum_{j\in S}m_{j}
  \end{align*}
  Consider the first inequality. The only binding cases are with $T=\emptyset$
  -- the left side is largest, and the right side is smallest. This
  yields $\sum_{i\in N,j\in S}\mu_{ij}\leq\sum_{i\in N}c_{j}$, recovering
  the capacity constraints in Equation \ref{eq:Cap-IV}. Now consider
  the second inequality; the only binding cases are $T=N$. This yields
  $\sum_{i\in N,j\in O\backslash S}\mu_{ij}\leq N-\sum_{j\in S}m_{j}$,
  recovering Equation \ref{eq:Min-IV}.
\end{proof}
\medskip{}

\begin{proof}[\textbf{Proof of Proposition \ref{prop:rand_mins_equiv}.}]
  Let $\mu\in[0,1]^{N\times O}$ satisfy unit demand (relaxed) and
  capacity constraints. It suffices to show that $\mu$ satisfies condition
  (\ref{eq:Min-III}) if and only if it satisfies condition (\ref{eq:Min-IV}),
  taking capacity and unit demand as given. For all $S\subseteq O_{m}$,
  condition (\ref{eq:Min-IV}) requires
  \begin{align*}
    \sum_{j\in O\backslash S}\mu_{j}                                 & \leq N-\sum_{j\in S}m_{j} &  & \tag{Min-III} \\
    \iff\sum_{j\in O}\mu_{j}-\sum_{j\in S}\mu_{j}+\sum_{j\in S}m_{j} & \leq N                                       \\
    \iff\sum_{j\in O}\mu_{j}+\sum_{j\in S}\left(m_{j}-\mu_{j}\right) & \leq N.
  \end{align*}
  Consider the terms $(m_{j}-\mu_{j})$ and the subset $S\subseteq O_{m}$.
  The above inequality is binding if and only if $\{j:m_{j}-\mu_{j}>0\}\subseteq S\subseteq\{j:m_{j}-\mu_{j}\geq0\}$.
  For all other $S$, this condition is redundant. Thus, we can equivalently
  write this condition as
  \begin{align*}
    \sum_{j\in O}\mu_{j}+\sum_{j\in O_{m}}\max\{m_{j}-\mu_{j},0\}                  & \leq N                    \\
    \iff\sum_{j\in O\setminus O_{m}}\mu_{j}+\sum_{j\in O_{m}}\max\{m_{j},\mu_{j}\} & \leq N                    \\
    \iff\sum_{j\in O}\max\{m_{j},\mu_{j}\}                                         & \leq N. &  & \tag{Min-IV}
  \end{align*}
\end{proof}

\subsection{Proof of Theorem \ref{thm:MPS}: Envy free}

\label{proof:envyfree}
\begin{proof}[\textbf{Proof that MPS is envy free.}]
  It is clear from the definition of our algorithm that it satisfies
  \textbf{common availability}; an object is always available to everyone
  or to no one. Formally, our capacity and minimum conditions can be
  written in terms of the marginals $\mu_{j}$ over agents. For $j\in O$,
  let $t_{j}\in[0,d]$ be the time at which object $j$ becomes unavailable
  during the process of MPS. (By common availability, $t_{j}$ is well
  defined and unique for each $j\in O$.) A formal definition of $t_{j}$
  is contained in the next section.

  Without loss of generality, consider agent 1, and label objects in
  $O$ such that $a\succ_{1}b\succ_{1}\cdots$. Since 1 consumes $a$
  over the time interval $[0,t_{a})$, his allocation $\mu_{a}$ satisfies
  $\mu_{1a}\geq\mu_{ia}$ for any $i\in N$. Now consider $\{a,b\}$.
  Agent 1 consumes from $\{a,b\}$ over the interval $[0,\max\{t_{a},t_{b}\})$,
  so $\mu_{1a}+\mu_{1b}\geq\mu_{ia}+\mu_{ib}$ for all $i\in N$. Iterating
  through all objects in $O$ gives the desired result.
\end{proof}

\subsection{Proof of Theorem \ref{thm:MPS}: Weak strategyproof}

\label{proof:wsp}

Without loss of generality suppose that agent 1 misreports his preferences
as $\succ_{1}'$, and let $\succ'=(\succ_{1}',\succ_{-1})$. Also
without loss of generality, let $a\in O$ be 1's favorite object.
Denote $\mu=MPS(\succ)$ and $\mu'=MPS(\succ')$.
\begin{defn}
  Let agent $i$'s \textbf{eating function} be $e_{i}:[0,1]\rightarrow O$,
  denoting his consumption throughout $MPS(\succ)$. Note that $e_{i}$
  is right-continuous; that is, for all $t\in[0,1)$, $\exists\varepsilon>0$
  such that $e_{i}(s)=e_{i}(t)$ for all $s\in[t,t+\varepsilon)$. Let
  $e=(e_{i})_{i\in N}$ denote an eating profile.

  Let $e_{i}'(\cdot)$ be $i$'s eating function under $MPS(\succ')$.
  Note we may have $e_{i}'(t)\neq e_{i}(t)$ for any $i\in N$, not
  just agent 1.

  Finally, let $n_{j}(t,e)=|\{i\in N:e_{i}(t)=j\}|$.
\end{defn}
\begin{defn}
  For any set of objects $S\subseteq O$, perhaps a singleton, let $n_{S}(t,e)$
  be the number of agents consuming $S$ at time $t$ under eating profile
  $e$. Denote
  \[
    \mu_{S}(t,e)=\int_{0}^{t}n_{S}(s,e)\ ds.
  \]
  to be the amount consumed from $S$ up to time $t$ under $e$.

  Similarly, for a set of agents $A\subseteq N$, let $n_{A,S}(t,e)$
  be the number of agents among $A$ consuming $S$ at time $t$ under
  eating profile $e$. Denote
  \[
    \mu_{A,S}(t,e)=\int_{0}^{t}n_{A,S}(s,e)\ ds.
  \]
\end{defn}
Recall from Proposition \ref{prop:rand_mins_equiv} that there are
two possible constraints that can bind on an object:
\begin{align*}
  \mu_{j}                            & \leq c_{j} \\
  \sum_{j\in O}\max\{m_{j},\mu_{j}\} & \leq N.
\end{align*}

\begin{defn}
  We refer to two conditions on an object $j$, time $t$, and eating
  schedule $e$
  \begin{align}
    \mu_{j}(t,e)                                                                                                     & <c_{j}\label{eq:cap}    \\
    \exists\varepsilon>0\text{ s.t. }\max\{m_{j},\mu_{j}(t,e)+\varepsilon\}+\sum_{k\neq j}\max\{m_{k},\mu_{k}(t,e)\} & \leq N\label{eq:global}
  \end{align}
  If condition (\ref{eq:cap}) fails for $j$ at $(t,e)$, we say that
  $j$ is bound by (\ref{eq:cap}) at $(t,e)$. That is, it is at its
  capacity. Likewise, if condition (\ref{eq:global}) fails for $j$
  at $(t,e)$, we say that $j$ is bound by (\ref{eq:global}) at $(t,e)$.
  We say $j$ is available at $(t,e)$ if it satisfies both (\ref{eq:cap})
  and (\ref{eq:global}) at $(t,e)$.

  Let
  \[
    \tau=\inf\{t:\exists j\in O\text{ s.t. }j\text{ fails }(\ref{eq:global})\text{ at }(t,e)\}.
  \]
  That is, $\tau$ is the time at which $\sum_{j\in O}\max\{m_{j},\mu_{j}\}=N$.
  After $\tau$, note that only objects such that $\mu_{j}(t,e)<m_{j}$
  can be consumed further. Let $\tau'$ be the analogous time under
  the eating profile $e'$.

  Finally, let
  \[
    t_{j}=\inf\{t:j\text{ fails }(\ref{eq:cap})\text{ or }(\ref{eq:global})\text{ at }(t,e)\}.
  \]
  That is, $t_{j}$ is the closing time of object $j$ under $e$; no
  more can be consumed of $j$ after $t_{j}$. Let $t_{j}'$ be the
  analogous time under the eating profile $e'$.
\end{defn}
Suppose that for an eating profile $e$, at some time $t\in[0,1]$
there exists an object $j$ that fails condition (\ref{eq:global}).
That is, $\mu_{j}(t,e)+\sum_{k\neq j}\max\{m_{k},\mu_{k}(t,e)\}=N$.
Note that $t\geq\tau$. Define $S:=\{j\in O:j\text{ fails (\ref{eq:cap}) or (\ref{eq:global}) at }(t,e)\}$.
Note then that $O\backslash S=\{j:\mu_{j}(t,e)<m_{j}\}$. Thus $\mu_{S}(t,e)$
is an upper bound on $\mu_{S}$ at any time and under any eating function,
since we have $\mu_{S}(t,e)=N-\sum_{j\notin S}m_{j}$.
\begin{defn}
  Let $\delta(t)$ be the amount of time from $[0,t)$ which agent 1
  has spent eating differently between the reports:
  \[
    \delta(t)=\int_{0}^{t}1\{e_{1}(s)\neq e_{1}'(s)\}\ ds.
  \]
\end{defn}
\begin{lem}
  \label{lem:sp1} Suppose $t_{a}\geq\tau$ and $\tau'\geq\tau$. Then
  for any object $b\neq a$ we have
  \[
    \mu_{b}(\tau,e')\geq\mu_{b}(\tau,e).
  \]
\end{lem}
\begin{proof}[Proof of Lemma \ref{lem:sp1}]
  If $\mu_{b}(\tau,e')=c_{b}$, then the claim is immediate. Now let
  $\mu_{b}(\tau,e')<c_{b}$. Toward a contradiction, suppose $\mu_{b}(\tau,e')<\mu_{b}(\tau,e)$
  for some $b\neq a$. Then there must exist some agent $i$ and time
  $t<\tau$ such that $e_{i}(t)=b$ but $e'_{i}(t)=c\neq b$. If not,
  then $n_{b}(s,e)\leq n_{b}(s,e')$ for all $s<\tau$. Then
  \[
    \mu_{b}(\tau,e)=\int_{0}^{\tau}n_{b}(s,e)\ ds\leq\int_{0}^{\tau}n_{b}(s,e')\ ds=\mu_{b}(\tau,e'),
  \]
  a contradiction. Note $i\neq1$ since $e_{1}(s)=a$ for all $s<t_{a}$.

  Since $t<\tau\leq\tau'$ and $\mu_{b}(t,e')\leq\mu_{b}(\tau,e')<c_{b}$,
  object $b$ is available at $(t,e')$. Then since $i$ eats truthfully
  and is consuming $c$ at $(t,e')$, $cP_{i}b$ and $t_{c}\leq t<t'_{c}$.
  Let $S:=\{j\in O:t_{j}<t'_{j}\}$, which is nonempty since $c\in S$.
  Let $d\in S$ be the element of $S$ such that $t_{d}$ is minimal.
  Note that $t_{d}\leq t_{c}<\tau$ so $\mu_{d}(t_{d},e)=c_{d}$. Note
  also $d\neq a$ since $t_{d}<\tau\leq t_{a}$.

  Since $t_{d}<t'_{d}$, there must exist some agent $\ell$ and time
  $r<t_{d}$ such that $e_{\ell}(r)=d$ but $e'_{\ell}(r)=x\neq d$.
  If not, then $n_{d}(s,e)\leq n_{d}(s,e')$ for all $s<t_{d}$. But
  this would imply that
  \begin{align*}
    c_{d} & =\mu_{d}(t_{d},e)                                                     \\
          & =\int_{0}^{t_{d}}n_{d}(s,e)\ ds                                       \\
          & \leq\int_{0}^{t_{d}}n_{d}(s,e')\ ds                                   \\
          & <\int_{0}^{t_{d}}n_{d}(s,e')\ ds+\int_{t_{d}}^{t'_{d}}n_{d}(s,e')\ ds \\
          & =\mu_{d}(t'_{d},e'),
  \end{align*}
  which is not possible. The strict inequality in the fourth line follows
  from the fact that $n_{d}(\cdot,e')$ is non-decreasing.

  Note that $\ell\neq1$, since $e_{1}(s)=a$ for all $s<t_{a}$. Since
  $r<t_{d}<t'_{d}$, object $d$ is available at $(r,e')$. Since $\ell$
  eats truthfully and is consuming $x$ at $(r,e')$, $xP_{\ell}d$
  and $t_{x}\leq r<t'_{x}$. Then $x\in S$ as well. However, $t_{x}\leq r<t_{d}$,
  contradicting that $t_{d}$ was minimal.
\end{proof}
\begin{lem}
  \label{lem:sp2} If $t_{a}>\tau$, then $\tau\geq\tau'$.
\end{lem}
\begin{proof}[Proof of Lemma \ref{lem:sp2}]
  This is a corollary of Lemma \ref{lem:sp1}. Suppose $t_{a}>\tau$.
  Let $S:=\{j\in O:t_{j}\leq\tau\}$. Note $a\notin S$. By Lemma \ref{lem:sp1},
  $\mu_{j}(\tau,e')\geq\mu_{j}(\tau,e)$ for all $j\in S$. Thus
  \begin{align*}
    \mu_{S}(\tau,e') & \geq\mu_{S}(\tau,e)       \\
                     & =N-\sum_{j\notin S}m_{j}.
  \end{align*}
  Thus condition (\ref{eq:global}) must bind under $e'$ weakly before
  time $\tau$. So $\tau'\leq\tau$ as desired.
\end{proof}
\begin{lem}
  \label{lem:sp3} Denote $\mu_{-1,a}(t,e)=\sum_{i\neq1}\mu_{i,a}(t,e)$.
  Let $t_{a}<t'_{a}$. Then
  \[
    \mu_{-1,a}(t_{a},e)\leq\mu_{-1,a}(t_{a},e').
  \]
\end{lem}
\begin{proof}[Proof of Lemma \ref{lem:sp3}]
  Denote $n_{-1,j}(t,e)=|\{i\in N\backslash\{1\}:e_{i}(t)=j\}|$. Toward
  a contradiction, suppose $\mu_{-1,a}(t_{a},e)>\mu_{-1,a}(t_{a},e')$.
  Then there must be some agent $i\neq1$ and time $t<t_{a}$ such that
  $e_{i}(t)=a$ but $e'_{i}(t)=b\neq a$. If not, then $n_{-1,a}(s,e)\leq n_{-1,a}(s,e')$
  for all $s<t_{a}$, and
  \[
    \mu_{-1,a}(t_{a},e)=\int_{0}^{t_{a}}n_{-1,a}(s,e)\ ds\leq\int_{0}^{t_{a}}n_{-1,a}(s,e')\ ds=\mu_{-1,a}(t_{a},e'),
  \]
  a contradiction.

  Since $t<t_{a}<t'_{a}$, $a$ is available at $(t,e')$. Then since
  $i$ eats truthfully, $bP_{i}a$ and $t_{b}\leq t<t'_{b}$. Let $S:=\{j\in O:t_{j}<t'_{j}\}$,
  which is nonempty since $b\in S$. Let $c\in S$ be the element of
  $S$ such that $t_{c}$ is minimal. Note that $t_{c}\leq t_{b}\leq t<t_{a}$
  and $t_{a}<1$. Since $t_{c}<t_{a}$ we know $c\neq a$.

  Since $t_{c}<t'_{c}$, there is some time $r<t_{c}$ and agent $\ell$
  such that $e_{\ell}(r)=c$ but $e'_{\ell}(r)=d\neq c$. If not, then
  $n_{c}(s,e)\leq n_{c}(s,e')$ for all $s<t_{c}$, and
  \begin{align*}
    \mu_{c}(t_{c},e) & =\int_{0}^{t_{c}}n_{c}(s,e)\ ds                                        \\
                     & \leq\int_{0}^{t_{c}}n_{c}(s,e')\ ds                                    \\
                     & \leq\int_{0}^{t_{c}}n_{c}(s,e)\ ds+\int_{t_{c}}^{t'_{c}}n_{c}(s,e)\ ds \\
                     & =\mu_{c}(t'_{c},e').
  \end{align*}
  In fact, we show that the third line is strict. It suffices to show
  that $n_{c}(s,e)\geq1$ for some $s<t_{c}$. Suppose not; then $\mu_{c}(t_{c},e)=0$.
  Since $t_{c}<1$, it must be that $m_{c}=0$ and thus both $t_{c}=\tau$
  and $t'_{c}=\tau'$. Since $\tau=t_{c}<t_{a}$, we know that $\tau<t_{a}$.
  Additionally, $t_{c}<t'_{c}$ gives $\tau<\tau'$. But this contradicts
  Lemma \ref{lem:sp2}. Thus, $\mu_{c}(t_{c},e)<\mu_{c}(t'_{c},e')$.

  If $t_{c}<\tau$, then $\mu_{c}(t_{c},e)=c_{c}$. But then $c_{c}<\mu_{c}(t'_{c},e')$,
  a contradiction. Now suppose $t_{c}\geq\tau$. Since $\mu_{c}(t_{c},e)\geq m_{c}$,
  we know that $\mu_{c}(t'_{c},e')>m_{c}$. Then $t'_{c}\leq\tau'$.
  But then we have $t_{a}>t_{c}\geq\tau$ and $\tau'\geq t'_{c}>t_{c}\geq\tau$.
  But this contradicts Lemma \ref{lem:sp2}. Thus we have shown that
  there is some time $r<t_{c}$ and agent $\ell$ such that $e_{\ell}(r)=c$
  but $e'_{\ell}(r)=d\neq c$.

  Note that $\ell\neq1$ since $t_{c}<t_{a}$ and $e_{1}(s)=a$ for
  all $s<t_{a}$. Since $r<t_{c}<t'_{c}$, $c$ is available at $(r,e')$.
  Then since $\ell$ eats truthfully, $dP_{\ell}c$ and $t_{d}\leq r<t'_{d}$.
  Then $d\in S$ as well. But $t_{d}\leq r<t_{c}$, contradicting that
  $t_{c}$ was minimal.
\end{proof}
\begin{proof}[\textbf{Proof that MPS is weak strategyproof.}]
  The objective will be to show that if $\delta(t_{a})>0$, then $\mu_{1a}'<\mu_{1a}$.
  The same argument can then be repeated for other objects besides $a$.

  Suppose that $\mu_{1a}'\geq\mu_{1a}$. Suppose also $\delta(t_{a})>0$,
  otherwise there is no change. Then $t_{a}'>t_{a}$, otherwise it could
  not be that $\mu_{1a}'\geq\mu_{1a}$. Additionally, $t_{a}<1$, or
  else $\mu_{1a}=1$ and no improvement is possible.

  It suffices to show that $\mu_{1,a}(t_{a}',e')-\mu_{1,a}(t_{a},e')<\delta(t_{a})$
  for the following reason. By definition, $\mu_{1,a}(t_{a},e')=t_{a}-\delta(t_{a})$.
  Therefore, if $\mu_{1,a}(t_{a}',e')-\mu_{1,a}(t_{a},e')<\delta(t_{a})$,
  then
  \begin{align*}
    \mu_{1,a}(t'_{a},e') & =\mu_{1,a}(t_{a},e')+(\mu_{1,a}(t'_{a},e')-\mu_{1,a}(t_{a},e')) \\
                         & <\mu_{1,a}(t_{a},e')+\delta(t_{a})                              \\
                         & =t_{a}-\delta(t_{a})+\delta(t_{a})                              \\
                         & =t_{a}=\mu_{1,a}(t_{a},e),
  \end{align*}
  a contradiction.

  We will repeatedly use the following fact. Note that
  \begin{align*}
    \mu_{a}(t_{a},e)-\mu_{a}(t_{a},e') & =\mu_{-1,a}(t_{a},e)-\mu_{-1,a}(t_{a},e')  \\
                                       & \ +\mu_{1,a}(t_{a},e)-\mu_{1,a}(t_{a},e').
  \end{align*}
  By Lemma \ref{lem:sp3}, $\mu_{-1,a}(t_{a},e)\leq\mu_{-1,a}(t_{a},e')$.
  By definition, $\mu_{1,a}(t_{a},e)-\mu_{1,a}(t_{a},e')=\delta(t)$.
  Thus
  \[
    \mu_{a}(t_{a},e)-\mu_{a}(t_{a},e')\leq\delta(t_{a}).
  \]

  We proceed by cases relating $t_{a}$ and $\tau$.

  \medskip{}
  \textit{Case 1.} $t_{a}<\tau$. \medskip{}

  \noindent Since $t_{a}<\tau$, we have $\mu_{a}(t_{a},e)=c_{a}$.
  Then
  \begin{align*}
    \mu_{a}(t_{a}',e')-\mu_{a}(t_{a},e') & \leq c_{a}-\mu_{a}(t_{a},e')        \\
                                         & =\mu_{a}(t_{a},e)-\mu_{a}(t_{a},e') \\
                                         & \leq\delta(t_{a}).
  \end{align*}
  Since $t_{a}<1$ and $\mu_{a}(t_{a},e)=c_{a}\geq1$, we know $\mu_{-1,a}(t_{a},e)>0$.
  Therefore, Lemma \ref{lem:sp3} implies $\mu_{-1,a}(t_{a},e')>0$.
  That is, at least one agent $i\neq1$ has consumed $a$ before time
  $t_{a}$ under eating profile $e'$. Since agent $i$ will continue
  to consume $a$ until it is closed, and $t'_{a}>t_{a}$, we know that
  $\mu_{-1,a}(t'_{a},e')-\mu_{-1,a}(t_{a},e')>0$. Since $\mu_{a}(t_{a}',e')-\mu_{a}(t_{a},e')\leq\delta(t_{a})$,
  this implies $\mu_{1,a}(t'_{a},e')-\mu_{1,a}(t_{a},e')<\delta(t_{a})$
  as desired.

  \medskip{}
  \textit{Case 2.} $t_{a}>\tau$. \medskip{}

  \noindent By Lemma \ref{lem:sp2}, $\tau'\leq\tau$. In summary, $\tau'\leq\tau<t_{a}<t_{a}'$.
  Since object $a$ closed after time $\tau$ under eating profile $e$
  and after time $\tau'$ under eating profile $e'$, we know $\mu_{a}(t_{a},e)=\mu_{a}(t'_{a},e')=m_{a}\geq1$.
  Then
  \begin{align*}
    \mu_{a}(t'_{a},e')-\mu_{a}(t_{a},e') & =m_{a}-\mu_{a}(t_{a},e')               \\
                                         & \leq\mu_{a}(t_{a},e)-\mu_{a}(t_{a},e') \\
                                         & \leq\delta(t_{a}).
  \end{align*}
  Since $t_{a}<1$ and $\mu_{a}(t_{a},e)=m_{a}\geq1$, we know $\mu_{-1,a}(t_{a},e)>0$.
  The rest of the argument follows Case 1.

  \medskip{}
  \textit{Case 3.} $t_{a}=\tau$. \medskip{}

  \noindent If $\tau'\leq\tau$, then $\tau'\leq\tau=t_{a}<t'_{a}$
  and so $\mu_{a}(t'_{a},e')=m_{a}\geq1$. The remainder of the argument
  is identical to Case 2.

  Instead, suppose $\tau'>\tau$. Let $S:=\{j\in O:t_{j}\leq\tau\}$;
  note that $a\in S$. Then $S\equiv\{j\in O:j\text{ fails (\ref{eq:cap}) or (\ref{eq:global}) at }(\tau,e)\}$,
  so $\mu_{S}(\tau,e)=N-\sum_{j\notin S}m_{j}$ is an upper bound for
  $\mu_{S}(t_{a}',e')$. Then
  \begin{align*}
    \mu_{S}(t'_{a},e')-\mu_{S}(\tau,e') & \leq\mu_{S}(\tau,e)-\mu_{S}(\tau,e').
  \end{align*}
  Recall $t_{a}=\tau$, so $\mu_{a}(t_{a},e)-\mu_{a}(t_{a},e')\leq\delta(t_{a})\Rightarrow\mu_{a}(\tau,e)-\delta(\tau)\leq\mu_{a}(\tau,e')$.
  Combined with the previous inequality, we have
  \begin{align*}
    \mu_{S}(t'_{a},e')-\mu_{S}(\tau,e') & \leq\mu_{S}(\tau,e)-\mu_{S}(\tau,e')                                              \\
                                        & =\mu_{S}(\tau,e)-\mu_{S\backslash\{a\}}(\tau,e')-\mu_{a}(\tau,e')                 \\
                                        & \leq\mu_{S}(\tau,e)-\mu_{S\backslash\{a\}}(\tau,e')-\mu_{a}(\tau,e)+\delta(\tau).
  \end{align*}
  By Lemma \ref{lem:sp1}, $\mu_{S\backslash\{a\}}(\tau,e')\geq\mu_{S\backslash\{a\}}(\tau,e)$,
  so
  \begin{align}
    \mu_{S}(t'_{a},e')-\mu_{S}(\tau,e') & \leq\mu_{S}(\tau,e)-\mu_{S\backslash\{a\}}(\tau,e')-\mu_{a}(\tau,e)+\delta(\tau)\nonumber     \\
                                        & \leq\mu_{S}(\tau,e)-\mu_{S\backslash\{a\}}(\tau,e)-\mu_{a}(\tau,e)+\delta(\tau)\label{eq:3_1} \\
                                        & =\delta(\tau).\nonumber
  \end{align}
  Since $t_{a}=\tau$, we have $\mu_{S}(t'_{a},e')-\mu_{S}(t_{a},e')\leq\delta(t_{a})$.
  In particular, since $a\in S$, $\mu_{a}(t'_{a},e')-\mu_{a}(t_{a},e')\leq\delta(t_{a})$.

  If $\mu_{-1,a}(t_{a},e)>0$, by Lemma \ref{lem:sp3}, we have $\mu_{-1,a}(t_{a},e')>0$.
  Agents consume an object until it closes, so our assumption that $t'_{a}>t_{a}$
  implies $\mu_{-1,a}(t_{a}',e')-\mu_{-1,a}(t_{a},e')>0$. Then since
  \begin{align*}
    \delta(t_{a}) & \geq\mu_{a}(t'_{a},e')-\mu_{a}(t_{a},e')                                                  \\
                  & =(\mu_{-1,a}(t_{a}',e')-\mu_{-1,a}(t_{a},e'))+(\mu_{1,a}(t_{a}',e')-\mu_{1,a}(t_{a},e')),
  \end{align*}
  we have $\mu_{1,a}(t_{a}',e')-\mu_{1,a}(t_{a},e')<\delta(t_{a})$
  as desired.

  Now suppose $\mu_{-1,a}(t_{a},e)=0$. Since $t_{a}=\tau<1$ and only
  agent 1 consumes $a$, we have $\mu_{a}(t_{a},e)<1$. Then $a$ closed
  at $t_{a}=\tau$ having met its minimum, so it must be that $m_{a}=0$.
  Therefore, $t'_{a}=\tau'$ also. Recall that $\mu_{S}(\tau,e)=N-\sum_{j\notin S}m_{j}$.
  Since $0<\mu_{a}(\tau,e)<1$ is not an integer, there is another object
  $b\in S$ such that $\mu_{b}(\tau,e)$ is not an integer. Since $t_{b}\leq\tau$,
  this implies $m_{b}<\mu_{b}(\tau,e)<c_{b}$.

  Suppose $t'_{b}\leq\tau$. Since $\tau'>\tau$, this implies $b$
  must have closed because it reached its capacity; i.e., $\mu_{b}(t'_{b},e')=c_{b}$.
  Then $\mu_{b}(t'_{b},e')=\mu_{b}(\tau,e')=c_{b}$, so $\mu_{b}(\tau,e')>\mu_{b}(\tau,e)$.
  Together with Lemma \ref{lem:sp1}, this implies
  \[
    \mu_{S\backslash\{a\}}(\tau,e')>\mu_{S\backslash\{a\}}(\tau,e).
  \]
  Substituting into Equation \ref{eq:3_1}, we have $\mu_{S}(t'_{a},e')-\mu_{S}(\tau,e')<\delta(\tau)$.
  Again, since $a\in S$ and $t_{a}=\tau$, it follows that $\mu_{a}(t'_{a},e')-\mu_{a}(t_{a},e')<\delta(t_{a})$
  and therefore
  \[
    \mu_{1,a}(t'_{a},e')-\mu_{1,a}(t_{a},e')<\delta(t_{a})
  \]
  as desired.

  Now suppose $t'_{b}>\tau$. By Lemma \ref{lem:sp1}, $\mu_{b}(\tau,e')\geq\mu_{b}(\tau,e)>0$.
  Thus there is an agent consuming $b$ at $(\tau,e')$. Since this
  agent will consume $b$ until $t'_{b}>\tau$ and since $t'_{a}>t_{a}=\tau$,
  we have
  \[
    \mu_{b}(t'_{a},e')-\mu_{b}(\tau,e')>0.
  \]
  By Equation \ref{eq:3_1}, we already have $\mu_{S}(t_{a}',e')-\mu_{S}(\tau,e')\leq\delta(\tau)$.
  Since $b\in S$ and $t_{a}=\tau$, the fact that $\mu_{b}(t'_{a},e')-\mu_{b}(\tau,e')>0$
  implies $\mu_{a}(t_{a}',e')-\mu_{a}(t_{a},e')<\delta(\tau).$ It follows
  that
  \[
    \mu_{1,a}(t'_{a},e')-\mu_{1,a}(t_{a},e')<\delta(t_{a})
  \]
  as desired.

  This completes the proof.
\end{proof}

\section{Relation to generalized probabilistic serial }

\label{sec:BCKM}

We first recount relevant definitions from BCKM13, though we refer
readers to the original paper for full details.
\begin{defn}
  \label{def:BCKM}Given $N\times O$ and a \textbf{constraint set}
  $S\subseteq N\times O$, a \textbf{constraint} is a restriction on
  deterministic allocations of the form
  \begin{align*}
    \underline{q}_{S}\leq\sum_{(i,j)\in S}M_{ij}\leq\bar{q}_{S}
  \end{align*}
  where $\underline{q}_{S},\bar{q}_{S}\in\mathbb{N}$. That is, a constraint
  imposes a floor and ceiling on a set of allocations between objects
  and agents. A \textbf{constraint structure} $\mathcal{H}$ is a set
  of constraint sets. A constraint structure $\mathcal{H}$ is a \textbf{hierarchy}
  if, for any $S,S'\in\mathcal{H}$, we have $S\subseteq S'$, $S'\subseteq S$,
  or $S\cap S'=\emptyset$. A constraint structure $\mathcal{H}$ is
  a \textbf{bihierarchy} if there exist hierarchies $\mathcal{H}_{1}$
  and $\mathcal{H}_{2}$ such that $\mathcal{H}=\mathcal{H}_{1}\cup\mathcal{H}_{2}$
  and $\mathcal{H}_{1}\cap\mathcal{H}_{2}=\emptyset$.
\end{defn}
Notice that the constraints in our model can be written as bihierarchical
constraints:
\begin{itemize}
  \item Capturing unit demand:
        \[
          1\leq\sum_{j\in O}M_{ij}\leq1
        \]
        The corresponding constraint sets are $\{i\}\times O$ for all $i\in N$,
        forming $\mathcal{H}_{1}$.
  \item Capturing minimums and capacities:
        \[
          m_{j}\leq\sum_{i\in N}M_{ij}\leq c_{j}
        \]
        The constraint sets are analogously $N\times\{j\}$ for all $j\in O$,
        forming $\mathcal{H}_{2}$.
\end{itemize}
Therefore, we can apply the implementation results contained in BCKM13.

BCKM13 also generalizes \textcite{BM01}'s Probabilistic Serial mechanism
to a more general environment within bihierarchies, called Generalized
Probabilistic Serial (GPS). GPS accepts constraints of the form
\begin{align}
  1 & \leq\sum_{(i,j)\in S}M_{ij}\leq1\quad\forall S=\{i\}\times O\nonumber                   \\
  0 & \leq\sum_{(i,j)\in S}M_{ij}\leq\bar{q}_{S}\quad\forall S\neq\{i\}\times O\label{eq:GPS}
\end{align}
where each $S\in\mathcal{H}$ and $\mathcal{H}$ is a bihierarchy.
The first line imposes unit demand, and the second allows only ceiling
constraints on remaining sets.

Even though our primitive constraints are bihierarchical, minimums
are not accepted by GPS. Further, even though we are able to re-express
our constraints as maximums over sets of objects in Proposition \ref{prop:rand_mins},
these constraint sets do not form a bihierarchy. These constraints
include ceilings on sets of the form $O\backslash\{j\}$, which in
general will intersect but not have a subset relation.

Of course, this leaves open the possibility that there may be \emph{some}
expression of the primitive constraints as ceilings on constraint
sets that form a bihierarchy. The next proposition shows this is not
the case, and therefore MPS is not a special case of GPS.
\begin{prop}
  \label{proof:bckmdiff} Let $|N|\geq3,|O|\geq3$. The constraints
  in the model (Section \ref{sec:Model}) cannot generally be expressed
  as bihierarchical constraints of the form in Equation \ref{eq:GPS}.
\end{prop}
\begin{proof}[Proof of Proposition \ref{proof:bckmdiff}]
  It suffices to show the claim for a particular market. Let $N=\{1,2,3\}$
  and $O=\{o_{1},o_{2},o_{3}\}$, with $c_{j}=3$ for all $j\in O$,
  and $m_{1}=1,m_{2}=1,m_{3}=0$. For a deterministic allocation $M$
  and constraint set $S\subseteq N\times O$, denote $M(S)=\sum_{(i,j)\in S}M_{ij}$.
  Without loss, assume that each $i\times O\in\mathcal{H}_{1}$.

  First, note that
  \[
    A=\begin{bmatrix}1 & 0 & 0 \\
               1 & 0 & 0 \\
               0 & 0 & 1
    \end{bmatrix}\text{ and }B=\begin{bmatrix}0 & 1 & 0 \\
               0 & 1 & 0 \\
               0 & 0 & 1
    \end{bmatrix}
  \]
  are not allowable deterministic allocations. In $A$, $o_{2}$'s minimum
  requirement is not satisfied, and in $B$, $o_{1}$'s minimum requirement
  is not satisfied. Conversely, note that
  \[
    C=\begin{bmatrix}1 & 0 & 0 \\
               0 & 1 & 0 \\
               0 & 0 & 1
    \end{bmatrix}
  \]
  is indeed an allowable deterministic allocation.

  Therefore, we know that there exists some constraint set $S\in\mathcal{H}$
  such that $A(S)>\bar{q}_{S}$ and some constraint set $S'\in\mathcal{H}$
  such that $B(S')>\bar{q}_{S'}$. However, since $C$ is allowable,
  we also know that $C(S)\leq\bar{q}_{S}$ and $C(S')\leq\bar{q}_{S'}$.

  Note that $(1,o_{1})\in S$. If not, then for the allowable deterministic
  allocation
  \[
    A'=\begin{bmatrix}0 & 1 & 0 \\
               1 & 0 & 0 \\
               0 & 0 & 1
    \end{bmatrix}
  \]
  we would have $A'(S)\geq A(S)>\bar{q}_{S}$, a contradiction. By similar
  arguments, it is easy to show that $(2,o_{1}),(3,o_{3})\in S$ and
  that $(1,o_{2}),(2,o_{2}),(3,o_{3})\in S'$. Note that $(3,o_{3})\in S\cap S'$,
  so $S\cap S'\neq\emptyset$.

  But note that $(1,o_{1})\notin S'$. If it were, then $C(S')=B(S')>\bar{q}_{S'}$,
  contradicting that $C$ is allowable. Therefore $S\nsubseteq S'$.
  Also, note that $(2,o_{2})\notin S$. If it were, then $C(S)=A(S)>\bar{q}_{S}$,
  again contradicting that $C$ is allowable. Therefore, $S'\nsubseteq S$.

  Therefore, $S$ and $S'$ cannot be part of the same hierarchy. But
  also, neither $S$ nor $S'$ can be in $\mathcal{H}_{1}$. Thus, $\mathcal{H}$
  cannot be a bihierarchy.
\end{proof}

\end{document}